\documentclass{llncs}[11pt]
\usepackage{amssymb,amsmath,graphics,color}
\usepackage{graphicx}
\usepackage{multicol}
\usepackage{color}
\usepackage[noadjust]{cite}
\usepackage{enumitem}
\setlist{nolistsep}
\usepackage{hyperref}
\usepackage{comment}
\usepackage{pgf,tikz}
\usetikzlibrary{arrows,shapes}

\newtheorem{conjecture1}{Conjecture}

\usepackage{times}

\date{}
\author{Carl Feghali, Matthew Johnson, Dani\"el Paulusma
\thanks{Author supported by EPSRC (EP/K025090/1).}
}
\institute{
School of Engineering and  Computing Sciences, Durham University,\\
Science Laboratories, South Road,
Durham DH1 3LE, United Kingdom
\texttt{\{carl.feghali,matthew.johnson2,daniel.paulusma\}@durham.ac.uk}
}

\title{Kempe Equivalence of Colourings of Cubic Graphs}

\begin{document}
\maketitle

\begin{abstract}
Given a graph $G=(V,E)$ and a proper vertex colouring of $G$, a Kempe chain is a subset of~$V$ that induces a maximal connected subgraph of~$G$ in which every vertex has one of two colours.  To make a Kempe change is to obtain one colouring from another by exchanging the colours of vertices in a Kempe chain.  Two colourings are Kempe equivalent if each can be obtained from the other by a series of Kempe changes. A conjecture of Mohar asserts that, for $k \geq 3$, all $k$-colourings of $k$-regular graphs  that are not complete are Kempe equivalent. We address the case $k=3$ by showing that all $3$-colourings of a cubic graph $G$ are Kempe equivalent unless $G$ is the complete graph $K_4$ or the triangular prism. 
\end{abstract}

\section{Introduction}

Let $G = (V, E)$ denote a simple undirected graph and let $k$ be a positive integer. A \emph{$k$-colouring} of~$G$ is a mapping $\phi: V \rightarrow \{1, \dots, k\}$ such that $\phi(u) \not= \phi(v)$ if $uv \in E$. The \emph{chromatic number} of $G$, denoted by $\chi(G)$, is the smallest $k$ such that $G$ has a $k$-colouring. 

If $a$ and $b$ are distinct colours, then $G(a, b)$ denotes the subgraph of $G$ induced by vertices with colour $a$ or $b$. An \emph{$(a, b)$-component} of $G$ is a connected component of $G(a, b)$ and is known as a \emph{Kempe chain}.  A  \emph{Kempe change} is the operation of interchanging the colours of some $(a, b)$-component of $G$.  Let $C_k(G)$ be the set of all $k$-colourings of $G$. Two colourings $\alpha, \beta \in C_k(G)$ are \emph{Kempe equivalent}, denoted by $\alpha \sim_k \beta$, if each can be obtained from the other by a series of Kempe changes. The equivalence classes $C_k(G)/\sim_k$ are called \emph{Kempe classes}. 
 
Kempe changes were first introduced by Kempe in his well-known failed attempt at proving the Four-Colour Theorem.  The Kempe change method has proved to be a powerful tool with  applications to several areas such as timetables~\cite{rolf}, theoretical physics~\cite{wang1, wang2}, and Markov chains~\cite{vigoda}. The reader is referred to~\cite{mohar1, sokal} for further details. From a theoretical viewpoint, Kempe equivalence was first addressed by Fisk~\cite{fisk} who proved that all $4$-colourings of an Eulerian triangulation of the plane are Kempe equivalent. This result was later extended by Meyniel~\cite{meyniel1} who showed that all $5$-colourings of a planar graph are Kempe equivalent, and by Mohar~\cite{mohar1} who proved that all $k$-colourings, $k > \chi(G)$, of a planar graph $G$ are Kempe equivalent. Las Vergnas and Meyniel~\cite{meyniel3} extended Meyniel's result by proving that all $5$-colourings of a $K_5$-minor free graph are Kempe equivalent. Bertschi~\cite{marc} also showed that all $k$-colourings of a perfectly contractile graph are Kempe equivalent, thus answering a conjecture of Meyniel~\cite{meyniel2}.  We note that Kempe equivalence with respect to edge-colourings has also been investigated~\cite{mohar1, mohar2, ruth}.
 
Here we are concerned with a conjecture of Mohar~\cite{mohar1} on {\it $k$-regular} graphs, that is, graphs in which every vertex has degree $k$ for some $k\geq 0$. Note that, for every 2-regular graph $G$ that is not an odd cycle, it holds that $C_2(G)$ is a Kempe class.  Mohar conjectured the following (where $K_{k+1}$ is the complete graph on $k+1$ vertices).

\begin{conjecture1}[\cite{mohar1}]\label{c-1}
Let $k\geq 3$. If $G$ is a $k$-regular graph that is not $K_{k+1}$ then $C_k(G)$ is a Kempe class.
\end{conjecture1}


\tikzstyle{vertex}=[circle,draw=black, fill=black, minimum size=5pt, inner sep=1pt]
\tikzstyle{edge} =[draw,-,black,>=triangle 90]

\begin{figure}
\begin{center}
\begin{tikzpicture}[scale=0.75]

   \foreach \pos/\name  in {{(0,2)/p1}, {(4.55,2)/p2}, {(4.55,0)/p3}, {(0,0)/p4}, {(3.14,1)/p5}, {(1.41,1)/p6}}
        \node[vertex] (\name) at \pos {};
\foreach \source/ \dest  in {p1/p2,  p2/p3,p3/p4, p2/p5, p3/p5, p5/p6, p1/p6, p1/p4, p4/p6}
       \path[edge, black!50!white,  thick] (\source) --  (\dest);

\end{tikzpicture}
\end{center}
\caption{The 3-prism.}\label{fig:0}
\end{figure}
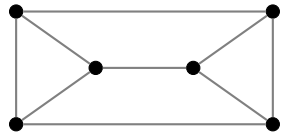

\noindent We address Conjecture~\ref{c-1} for the case $k = 3$. For this case the conjecture is known to be false. A counter-example is the 3-prism displayed in Figure~\ref{fig:0}. The fact that some $3$-colourings of the $3$-prism  are not Kempe equivalent was already observed by van den Heuvel~\cite{heuvel}.  Our contribution is that the 3-prism is the {\it only} counter-example for the case $k=3$, that is, we completely settle the case $k=3$ by proving  the following result for 3-regular graphs also known as {\it cubic} graphs.

\begin{theorem}\label{mainthm}
If $G$ is a cubic graph that is neither $K_4$ nor the $3$-prism then $C_3(G)$ is a Kempe class.
\end{theorem}

\noindent
We give the proof of our result in the next section. Besides exploiting the 3-regularity, our proof also takes into account the fact that one additional graph, namely the 3-prism, is forbidden. We did not find any counter-examples for $k\geq 4$ and believe Conjecture~\ref{c-1} may well hold for $k\geq 4$. As such, new techniques are necessary to tackle the remaining cases.

Our result is an example of a type of result that has received much recent attention: that of determining the structure of a \emph{reconfiguration graph}.  A reconfiguration graph has as vertex set all solutions to a search problem and an edge relation that describes a transformation of one solution into another.  Thus Theorem~\ref{mainthm} is concerned with the reconfiguration graph of 3-colourings of cubic graph with edge relation $\sim_k$ and shows that it is connected except in two cases. To date the stucture of reconfiguration graphs of colourings has focussed~\cite{BB13, BJLPP14, BP14, CHJ06, CHJ06b, FJP14} on the case where vertices are joined by an edge only when they differ on just one colour (that is, when one colouring can be transformed into another by a Kempe change of a Kempe chain that contains only one vertex).  For a survey of recent results on reconfiguration graphs see~\cite{heuvel}.
  
\section{The Proof of Theorem~\ref{mainthm}}\label{cubic}

We first give some further definitions and terminology. Let $G=(V,E)$ be a graph. Then~$G$ is {\it $H$-free} for some graph~$H$ if $G$ does not contain an induced subgraph isomorphic to $H$.  A \emph{separator} of  $G$ is a set $S \subset V$ such that $G - S$ has more components than~$G$. We say that $G$ is \emph{$p$-connected} for some integer~$p$ if $|V|\geq p+1$ and every separator of $G$ has size at least~$p$.  Some small graphs that we will refer to are defined by their illustrations in Figure~\ref{fig:1}.

\tikzstyle{vertex}=[circle,draw=black, fill=black, minimum size=5pt, inner sep=1pt]
\tikzstyle{edge} =[draw,-,black,>=triangle 90]

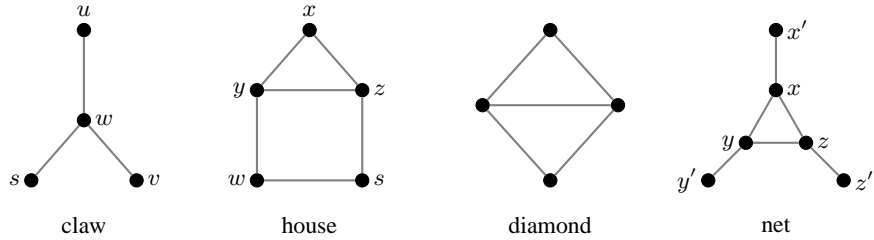
\begin{figure}
\begin{center}
\begin{tikzpicture}[scale=1]
   \foreach \pos/\name / \label / \posn / \dist in {{(1,0)/c1/u/above/2}, {(1,-1.2)/c2/w/{right}/1}, {(1.7,-2)/c3/v/{right}/1}, {(0.3,-2)/c4/s/{left}/1}}
       { \node[vertex] (\name) at \pos {};
       \node [\posn=\dist] at (\name) {$\label$};
       }
\foreach \source/ \dest  in {c2/c1,  c2/c4, c2/c3}
       \path[edge, black!50!white,  thick] (\source) --  (\dest);
       
   \foreach \pos/\name / \label / \posn / \dist  in {{(4,0)/h1/x/{above}/2}, {(3.3,-0.8)/h2/y/{left}/1}, 
   {(4.7,-0.8)/h3/z/{right}/1}, {(3.3,-2)/h4/w/{left}/1},{(4.7,-2)/h5/s/{right}/1}}
      { \node[vertex] (\name) at \pos {};
       \node [\posn=\dist] at (\name) {$\label$};
       }
\foreach \source/ \dest  in {h1/h2,  h1/h3, h2/h3, h2/h4, h4/h5, h5/h3}
       \path[edge, black!50!white,  thick] (\source) --  (\dest);     

   \foreach \pos/\name / \label in {{(7.2,0)/n1/x}, {(7.2,-2)/n2/y}, {(6.3,-1)/n3/w}, {(8.1,-1)/n4/z}}
       { \node[vertex] (\name) at \pos {};}
       
\foreach \source/ \dest  in {n1/n3,  n1/n4, n2/n3, n2/n4, n3/n4}
       \path[edge, black!50!white,  thick] (\source) --  (\dest);

\node at (1,-2.6) { claw};
\node at (4,-2.6) { house};
\node at (7.2,-2.6) {diamond};

 \begin{scope}[xshift=3cm]
   \foreach \pos/\name / \label / \posn / \dist in {{(6.3,-2)/n1/y'/{left}/1}, {(6.8,-1.5)/n2/y/{left}/1}, {(7.2,-0.8)/n3/x/{right}/1}, {(7.2,0)/n4/x'/{right}/1}, {(7.6,-1.5)/n5/z/{right}/1}, {(8.1,-2)/n6/z'/{right}/1}}
       { \node[vertex] (\name) at \pos {};
       \node [\posn=\dist] at (\name) {$\label$};
       }
       
\foreach \source/ \dest  in {n1/n2,  n2/n3, n3/n4, n2/n5, n3/n5, n5/n6}
       \path[edge, black!50!white,  thick] (\source) --  (\dest);

       \node at (7.2,-2.6) { net};
       
\end{scope}
\end{tikzpicture}
\end{center}
\caption{A number of special graphs used in our paper.}\label{fig:1}
\end{figure}

Besides three new lemmas, we will need the aforementioned result of van den Heuvel, which follows from the fact that for the 3-prism $T$, the subgraphs$T(1, 2)$, $T(2, 3)$ and $T(1, 3)$ are connected so that the number of Kempe classes is equal to the number of different 3-colourings of $T$ up to colour permutation, which is two.

\begin{lemma}[\cite{heuvel}]\label{pprism}
If $G$ is the 3-prism then $C_3(G)$ consists of two Kempe classes.
\end{lemma}

\begin{lemma}\label{l1}
If $G$ is a cubic graph that is not $3$-connected then $C_3(G)$ is a Kempe~class.
\end{lemma}

\begin{lemma}\label{l2}
If $G$ is a $3$-connected cubic graph that is claw-free but that is neither $K_4$ nor the $3$-prism then $C_3(G)$ is a Kempe class. 
\end{lemma}

\begin{lemma}\label{l3}
If $G$ is a $3$-connected cubic graph that is not claw-free then $C_3(G)$ is a Kempe class.
\end{lemma}

Observe that Theorem~\ref{mainthm} follows from the above lemmas, which form a case distinction. Hence it suffices to prove Lemmas~\ref{l1}--\ref{l3}.  These proofs form the remainder of the paper.

\subsection{Proof of Lemma~\ref{l1}}\label{s-lemma1}

In order to prove Lemma~\ref{l1} we need three auxiliary results and one more definition: a graph $G$ is {\it $d$-degenerate} if every induced subgraph of $G$ has a vertex with degree at most $d$.

\begin{lemma}[\cite{meyniel3, mohar1}]\label{degprop1}
Let $d$ and $k$ be any two integers with $d\geq 0$ and $k\geq d+1$.  If $G$ is a $d$-degenerate graph then $C_k(G)$ is a Kempe class. 
\end{lemma}

\begin{lemma}[\cite{meyniel3}]\label{ldeg}
Let $k\geq 1$ be an integer. Let $G_1, G_2$ be two graphs such that $G_1 \cap G_2$ is complete. If both $C_k(G_1)$ and $C_k(G_2)$ are Kempe classes then $C_k(G_1 \cup G_2)$ is a Kempe class.
\end{lemma}

\begin{lemma}[\cite{mohar1}]\label{ledge}
Let $k\geq 1$ be an integer and let $G$ be a subgraph of a graph $G'$. Let~$c_1$ and $c_2$ be the restrictions, to $G$, of two $k$-colourings $c_1'$ and $c_2'$ of $G'$. If $c_1'$ and $c_2'$ are Kempe equivalent then $c_1$ and $c_2$ are Kempe equivalent.
\end{lemma}

For convenience we restate Lemma~\ref{l1} before we present its proof.

\medskip
\noindent
{\bf Lemma~\ref{l1} (restated).}
{\it If $G$ is a cubic graph that is not $3$-connected then $C_3(G)$ is a Kempe~class.}

\medskip
\noindent {\it Proof.}
As disconnected graphs can be considered component-wise, we assume that $G$ is connected. As $G$ is cubic, $G$ has at least four vertices. Because $G$ is not 3-connected, $G$ has a separator $S$ of size at most~2. Let $S$ be a minimum separator of $G$ such that $G =G_1 \cup G_2$ and $G_1 \cap G_2=S$.  As every vertex in $S$ has degree at most~2 in each $G_i$ and $G$ is cubic, $G_1$ and $G_2$ are $2$-degenerate. Hence, by Lemma~\ref{degprop1}, $C_3(G_1)$ and $C_3(G_2)$ are Kempe classes.  If $S$ is a clique, we apply Lemma~\ref{ldeg}.  Thus we assume that $S$, and any other minimum separator of $G$, is not a clique. Then $S=\{x, y\}$ for two distinct vertices $x$ and $y$ with $xy \not\in E(G)$.
   
Because $S$ is a minimum separator, $x$ and $y$ are non-adjacent and $G$ is cubic, $x$ has either one neighbour in $G_1$ and two in $G_2$, or the other way around; the same holds for vertex~$y$. For $i=1,2$, let $N_i(x)$ and $N_i(y)$ be the set of neighbours of $x$ and~$y$, respectively, in $G_i$.  Then we have that either $|N_1(x)|=1$ and $|N_2(x)|=2$, or $|N_1(x)|=2$ and $|N_2(x)|=1$, and similarly, that either $|N_1(y)|=1$ and $|N_2(y)|=2$, or $|N_1(y)|=2$ and $|N_2(y)|=1$.  Let $x_1\in N_1(x)$ for some $x_1\in V(G_1)$.

We may assume that $|N_1(x)|\neq |N_1(y)|$; if not we can do as follows.  Assume without loss of generality that $N_1(x)=\{x_1\}$ and that $|N_1(y)|=1$. Then $\{x_1,y\}$ is a separator. By our assumption that $G$ has no minimum separator that is a clique, we find that $\{x_1,y\}$ is a minimum separator with $x_1y\notin E(G)$. As $G$ is cubic, $x_1$ has two neighbours in $V(G_1)\setminus \{x,x_1\}$. As $|N_1(y)|=1$ and $x_1$ and $y$ are not adjacent,~$y$ has exactly one neighbour in $V(G_1)\setminus \{x,x_1\}$.  Hence we could take $\{x_1,y\}$ as our minimum separator instead of~$S$ in order to get the desired property.  We may thus assume that $|N_1(x)|\not=|N_1(y)|$. As this means that $|N_2(x)|\not=|N_2(y)|$, we can let $N_1(x) = \{x_1\}$ and $N_2(y)=\{y_1\}$ for some $y_1\in V(G_2)$.

It now suffices to prove the following two claims.

\medskip
\noindent
{\it Claim 1. All colourings $\alpha$ such that $\alpha(x) \not=\alpha(y)$ are Kempe equivalent in $C_3(G)$. }
  
\medskip
\noindent
We prove Claim~1 as follows. We add an edge $e$ between $x$ and $y$. This results in graphs $G_1+e$, $G_2+e$ and $G+e$. We first prove that $C_3(G+e)$ is a Kempe class. Because $x$ and $y$ have degree~1 in $G_1$ and $G_2$, respectively, and $G$ is cubic, we find that the graphs $G_1 +e$ and $G_2 +e$ are $2$-degenerate. Hence, by Lemma~\ref{degprop1},  $C_3(G_1+e)$ and $C_3(G_2+e)$ are Kempe classes. By Lemma~\ref{ldeg}, it holds that $C_3(G+e)$ is a Kempe class. Applying Lemma~\ref{ledge} completes the proof of Claim~1.

\medskip
\noindent
{\it Claim 2. For every colouring $\alpha$ such that $\alpha(x) = \alpha(y)$, there exists a colouring $\beta$ with $\beta(x) \not= \beta(y)$ such that $\alpha$ and $\beta$ are Kempe equivalent in $C_3(G)$.}

\medskip
\noindent
We assume without loss of generality that $\alpha(x) = \alpha(y) = 1$ and $\alpha(y_1) = 2$.  If $\alpha(x_1)=2$ then we apply a Kempe change on the $(1,3)$-component of $G$ that contains $x$. Note that $y$ does not belong to this component. Hence afterwards we obtain the desired colouring $\gamma$. If $\alpha(x_1)=3$ then we first apply a Kempe change on the $(2,3)$-component of $G$ that contains $x_1$. Note that this does not affect the colours of $x$, $y$ and~$y_1$ as they do not belong to this component.  Afterwards we proceed as before.  This completes the proof of Claim~2 (and the lemma).\qed

\subsection{Proof of Lemma~\ref{l2}}\label{s-l2}
  
We require some further terminology and three lemmas.  We {\it identify} two vertices $x$ and~$y$ in a graph~$G$ if we replace them by a new vertex adjacent to all neighbours of $x$ and $y$ in~$G$. Two colourings $\alpha$ and $\beta$ of a graph $G$ \emph{match} if there exists two vertices $x, y$ with a common neighbour in $G$ such that $\alpha(x) = \alpha(y)$ and $\beta(x) = \beta(y)$. 

\begin{lemma}\label{lplanar}
Let $k\geq 1$ and $G'$ be the graph obtained from a graph $G$ by identifying two non-adjacent vertices $x$ and $y$. If $C_k(G')$ is a Kempe class then all $k$-colourings $c$ of $G$ with $c(x) = c(y)$ are Kempe equivalent. 
\end{lemma}

\begin{proof}
Let $\alpha$ and $\beta$ be two $k$-colourings of $G$ with $\alpha(x)=\alpha(y)$ and $\beta(x)=\beta(y)$. Let~$z$ be the vertex of $G'$ that is obtained after identifying $x$ and $y$.  Let $\alpha'$ and $\beta'$ be the $k$-colourings of $G'$ that agree with $\alpha$ and $\beta$, respectively, on $V(G)\setminus \{x,y\}$ and for which $\alpha'(z)=\alpha(x)(=\alpha(y))$  and $\beta'(z)=\beta(x)(=\beta(y))$. By our assumption, there exists a Kempe chain from $\alpha'$ to $\beta'$ in $G'$. We mimic this Kempe chain in $G$. Note that any $(a,b)$-component in $G'$ that contains $z$ corresponds to at most two $(a,b)$-components in~$G$, as $x$ and $y$ may get separated. Hence, every Kempe change on an $(a,b)$-component corresponds to either one or two Kempe changes in $G$ (if $x$ and $y$ are in different $(a,b)$-components  then we apply the corresponding Kempe change in $G'$ on each of these two components). In this way we obtain a Kempe chain from $\alpha$ to $\beta$ as required. \qed
\end{proof}

\begin{lemma}\label{lidentify1}
Let $k\geq 3$. If $\alpha$ and $\beta$ are matching $k$-colourings of a 3-connected graph $G$ of maximum degree~$k$ then $\alpha \sim_k \beta$. 
\end{lemma}

\begin{proof}
If $G$ is $(k-1)$-degenerate then $\alpha \sim_k \beta$ by Lemma~\ref{degprop1}. Assume that $G$ is not $(k-1)$-degenerate. Then $G$ is $k$-regular. Since $\alpha$ and $\beta$ match, there exist two vertices~$u$ and $v$ of $G$ that have a common neighbour $w$ such that $\alpha(u) = \alpha(v)$ and $\beta(u) = \beta(v)$.  Let $x$ denote the vertex of $G'$ obtained by identifying $u$ and $v$. 

Let $S$ be a separator of $G'$. If $S$ does not contain $x$ then $S$ is a separator of $G$. Then $|S|\geq 3$ as $G$ is 3-connected. If $S$ contains $x$ then $S$ must contain another vertex as well; otherwise $\{u,v\}$ is a separator of size~2 of~$G$, which is not possible. Hence, $|S|\geq 2$ in this case. We conclude that $G'$ is 2-connected.

We now prove that $G'$ is $(k-1)$-degenerate. Note that, in $G'$, $w$ has degree $k-1$, $x$ has degree at least~$k$ and all other vertices have degree $k$.   Let $u_1, \dots, u_r$ for some $r\geq k-1$ be the neighbours of $x$ not equal to $w$.  Since $G'$ is $2$-connected, the graph $G'' = G' \backslash x$ is connected. This means that every $u_i$ is connected to $w$ via a path in $G''$, which corresponds to a path in $G'$ that does not contain~$x$. Since $w$ has degree $k-1$ and every vertex not equal to $x$ has degree~$k$, we successively delete vertices of these paths starting from $w$ towards $u_i$ so that each time we delete a vertex of degree at most $k-1$. Afterwards we can delete $x$ as $x$ has degree~0. The remaining vertices form an induced subgraph of $G'$ whose components each have maximum degree at least~$k$ and at least one vertex of degree at most~$k-1$. Hence, we can continue deleting vertices of degree at most~$k-1$ and thus find that $G'$ is $(k-1)$-degenerate.  Then, by  Lemma~\ref{degprop1}, $C_k(G')$ is a Kempe class. Hence, by Lemma~\ref{lplanar}, we find that  $\alpha \sim_k \beta$ as required. This completes the proof. \qed
\end{proof}

\begin{lemma}\label{l-net}
Every 3-connected cubic claw-free graph $G$ that is neither $K_4$ nor the 3-prism is house-free, diamond-free and contains an induced net (see also Figure~\ref{fig:1}).
\end{lemma}

\begin{proof}
First suppose that $G$ contains an induced diamond $D$. Then, since $G$ is cubic, the two non-adjacent vertices in $D$ form a separator and $G$ is not 3-connected, a contradiction. Consequently, $G$ is diamond-free.

Now suppose that $G$ contains an induced house $H$. We use the vertex labels of Figure~\ref{fig:1}.  So, $s,w,x$ are the vertices that have degree~2 in $H$, and $s$ and $w$ are adjacent. As $G$ is cubic, $w$ has a neighbour $t\in V(G)\setminus V(H)$. Since $G$ is cubic and claw-free,~$t$ must be adjacent to $s$.  If $tx \in E$ then $G$ is the $3$-prism. If $tx \notin E$ then $t$ and $x$ form a separator of size~2. In either case we have a contradiction. Consequently, $G$ is house-free.

We now prove that $G$ has an induced net. As $G$ is cubic and claw-free, it has a triangle  and each vertex of the triangle has one neighbour in $G$ outside the triangle.  Because $G$ is not $K_4$ and diamond-free, these neighbours are distinct. Then, because~$G$ is house-free, no two of them are adjacent. Hence, together with the vertices of the triangle, they induce a net.\qed
\end{proof}

We restate Lemma~\ref{l2} before we present its proof.

\medskip
\noindent
{\bf Lemma~\ref{l2} (restated).}
{\it If $G$ is a $3$-connected cubic graph that is claw-free but that is neither $K_4$ nor the $3$-prism then $C_3(G)$ is a Kempe class.}

\medskip
\noindent
{\it Proof.} By Lemma~\ref{l-net}, $G$ contains an induced net~$N$. For the vertices of $N$ we use the labels of Figure~\ref{fig:1}. In particular, we refer to $x$, $y$ and $z$ as the \emph{t-vertices} of $N$, and $x'$, $y'$ and $z'$ as the  \emph{p-vertices}.  Let $\alpha$ and~$\beta$ be two $3$-colourings of~$G$. In order to show that $\alpha \sim_3 \beta$ we distinguish two cases. 
 
\medskip
\noindent
{\bf Case 1.} There are two p-vertices with identical colours under $\alpha$ or $\beta$.\\
Assume that $\alpha(x') = \alpha(y') = 1$. Then $\alpha(z) = 1$ as the t-vertices form a triangle, so colour 1 must be used on one of them.  Assume without loss of generality that $\alpha(z') = \alpha(x)=2$ and so $\alpha(y)=3$.  If  $\beta(z') = \beta(x)$ then $\alpha$  and $\beta$ match (as $x$ and $z'$ have $z$ as a common neighbour).  Then, by Lemma~\ref{lidentify1}, we find that $\alpha \sim_3 \beta$.  Otherwise $\beta(z') = \beta(y)$, since the colour of $z'$ must appear on one of $x$ and $y$.  Note that the $(2, 3)$-component containing $x$ under $\alpha$ consists only of $x$ and $y$. Then a Kempe exchange applied to this component yields a colouring $\alpha'$ such that $\alpha'(z') = \alpha'(y)$. As $y$ and $z'$ have $z$ as a common neighbour as well, this means that $\alpha'$ and $\beta$ match. Hence, it holds that $\alpha\sim_3 \alpha' \sim_3 \beta$, where the second equivalence follows from Lemma~\ref{lidentify1}.

 \medskip
 \noindent
 {\bf Case 2.} All three p-vertices have distinct colours under both $\alpha$ and $\beta$.\\
Assume without loss of generality that $\alpha(x) = \alpha(z') = 1$, $\alpha(y) = \alpha(x') = 2$, and $\alpha(z) = \alpha(y') = 3$.  Note that Kempe chains of $G$ are paths or cycles, as no vertex in a chain can have degree~3 since all its neighbours in a chain are coloured alike and $G$ is claw-free.  So, we will refer to \emph{$(a,b)$-paths} rather than $(a,b)$-components.

We will now prove that there exists a colouring $\alpha'$ with $\alpha \sim_3 \alpha'$ that assigns the same colour to two  p-vertices of~$N$. This suffices to complete the proof of the lemma, as afterwards we can apply Case~1.

Consider the $(1,2)$-path $P$ that contains $x'$.  If $P$ does not contain $z'$, then a Kempe exchange on $P$ gives us  a desired colouring $\alpha'$ (with $x'$ and $z'$ coloured alike). So we can assume that $x'$ and $z'$ are joined by a $(1,2)$-path $P_{12}$, and, similarly, $x'$ and $y'$ by a $(2,3)$-path $P_{23}$, and $y'$ and $z'$  by a $(1,3)$-path $P_{13}$. 

Let $G'$ be the subgraph of $G$ induced by the three paths. Note that $P_{12}$ has end-vertices $y$ and $z'$, $P_{23}$ has end-vertices $z$ and $x'$ and $P_{13}$ has end-vertices $x$ and $y'$. Hence, $G'$ contains the vertices of $N$ and every vertex in $G'-N$ is an internal vertex of one of the three paths. As $G$ is cubic, this means that each vertex in $G'-N$ belongs to exactly one path. Moreover, as $G$ is claw-free and cubic, two vertices in $G'-N$ that are on different paths are adjacent if and only if they have a $p$-vertex as a common neighbour.
 
In Figure~\ref{fig:4} are illustrations of $G'$ and the colourings of this proof that we are about to discuss. Let $x'' \neq x$ be the vertex in $P_{12}$ adjacent to $x'$. From the above it follows that~$x''$ is adjacent to the neighbour of $x'$ on $P_{23}$ and that no other vertex of $P_{12}$ (apart from $x'$) is adjacent to a vertex of $P_{23}$. As $G$ is cubic, this also means that $x''$ has no neighbour outside $G'$.  Apply a Kempe exchange on $P_{12}$ and call the resulting colouring~$\gamma$.  By the arguments above, the new $(2,3)$-path $Q_{23}$ (under $\gamma$) that contains $y'$ has vertex set $(V(P_{23}) \cup \{x''\}) \backslash \{x', y, z\}$.  Apply a Kempe exchange on $Q_{23}$. This results in a colouring $\alpha'$ with $\alpha'(y')=\alpha'(z')=2$, hence $\alpha'$ is a desired colouring. This completes the proof of Case~2 and thus of the lemma. \qed

\tikzstyle{vertex}=[circle,draw=black, fill=black, minimum size=5pt, inner sep=1pt]
\tikzstyle{edge} =[draw,-,black,>=triangle 90]

\begin{figure}
\begin{center}
\begin{tikzpicture}[scale=0.4]
   \foreach \pos/\name / \label / \posn / \dist /\labelb /\posnb /\distb in {{(0,0.85)/c1/x/{left}/1/1/right/1}, {(0,3.85)/c2/x'/{above}/1/2/{below left}/-1}, {(-2,3.85)/c3/{}/{left}/1/3/{below right}/-1}, {(2,3.85)/c4/x''/{above right}/-1/1/{below left}/-1}, {(-1,-0.85)/d1/y/{above left}/-1/2/below/1}, {(-3.68,-2.35)/d2/y'/{below left}/-1/3/{above}/1}, {(-4.68,-0.63)/d3/{}/{left}/1/2/left/1}, {(-2.68,-4.08)/d4/{}/{left}/1/1/{below left}/-1}, {(1,-0.85)/e1/z/{above right}/-1/3/below/1}, {(3.68,-2.35)/e2/z'/{below right}/-2/1/above/1}, {(4.68,-0.63)/e3/{}/{left}/1/2/right/1}, {(2.68,-4.08)/e4/{}/{left}/1/3/{below right}/-1}} 
       { \node[vertex] (\name) at \pos {};
       \node [\posn=\dist] at (\name) {$\label$};
       \node [\posnb=\distb] at (\name) {$\labelb$};
            }
\foreach \source/ \dest  in {c1/c2,  d1/d2,  e1/e2,  e1/c1, c1/d1, d1/e1, d2/d4, e2/e4, c2/c3, c2/c4, d2/d3, e2/e3}
       \path[edge, very thick] (\source) --  (\dest);

\draw[edge, very thick] (c3) .. controls (0,5.2)  ..  (c4);
\draw[edge, very thick] (d3) .. controls (-5.24,-3.35)  ..  (d4);
\draw[edge, very thick] (e3) .. controls (5.24,-3.35) .. (e4);
 
\draw[edge, thick, dashdotted] (c2) .. controls (-5, 4) and (-7.04, 3.2) ..  (d2) node [midway, above, sloped] (TextNode) {$(2,3)$-path};
\draw[edge, thick, dashdotted] (c2) .. controls (5, 4) and (7.04, 3.2) ..  (e2) node [midway, above, sloped] () {$(1,2)$-path};
\draw[edge, thick, dashdotted] (d2) .. controls (-0.32, -7.9) and (0.32, -7.9) ..  (e2) node [midway, below, sloped] () {$(1,3)$-path};

\node at (0,-8.7) {colouring $\alpha$};
\node at (16,-8.7) {colouring $\gamma$};
\node at (8,-24.5) {colouring $\alpha'$};

\begin{scope}[xshift=16cm, yshift=0cm]
   \foreach \pos/\name / \label / \posn / \dist /\labelb /\posnb /\distb in {{(0,0.85)/c1/x/{left}/1/2/right/1}, {(0,3.85)/c2/x'/{above}/1/1/{below left}/-1}, {(-2,3.85)/c3/{}/{left}/1/3/{below right}/-1}, {(2,3.85)/c4/x''/{above right}/-1/2/{below left}/-1}, {(-1,-0.85)/d1/y/{above left}/-1/1/below/1}, {(-3.68,-2.35)/d2/y'/{below left}/-1/3/{above}/1}, {(-4.68,-0.63)/d3/{}/{left}/1/2/left/1}, {(-2.68,-4.08)/d4/{}/{left}/1/1/{below left}/-1}, {(1,-0.85)/e1/z/{above right}/-1/3/below/1}, {(3.68,-2.35)/e2/z'/{below right}/-2/2/above/1}, {(4.68,-0.63)/e3/{}/{left}/1/1/right/1}, {(2.68,-4.08)/e4/{}/{left}/1/3/{below right}/-1}} 
       { \node[vertex] (\name) at \pos {};
       \node [\posn=\dist] at (\name) {$\label$};
       \node [\posnb=\distb] at (\name) {$\labelb$};
            }
\foreach \source/ \dest  in {c1/c2,  d1/d2,  e1/e2,  e1/c1, c1/d1, d1/e1, d2/d4, e2/e4, c2/c3, c2/c4, d2/d3, e2/e3}
       \path[edge, very thick] (\source) --  (\dest);

\draw[edge, very thick] (c3) .. controls (0,5.2)  ..  (c4);
\draw[edge, very thick] (d3) .. controls (-5.24,-3.35)  ..  (d4);
\draw[edge, very thick] (e3) .. controls (5.24,-3.35) .. (e4);
 
\draw[edge, thick, dashdotted] (c2) .. controls (-5, 4) and (-7.04, 3.2) ..  (d2) node [midway, above, sloped] (TextNode) {$(2,3)$-path};
\draw[edge, thick, dashdotted] (c2) .. controls (5, 4) and (7.04, 3.2) ..  (e2) node [midway, above, sloped] () {$(1,2)$-path};
\draw[edge, thick, dashdotted] (d2) .. controls (-0.32, -7.9) and (0.32, -7.9) ..  (e2) node [midway, below, sloped] () {$(1,3)$-path};

\end{scope}

\begin{scope}[xshift=8cm, yshift=-16cm]
   \foreach \pos/\name / \label / \posn / \dist /\labelb /\posnb /\distb in {{(0,0.85)/c1/x/{left}/1/2/right/1}, {(0,3.85)/c2/x'/{above}/1/1/{below left}/-1}, {(-2,3.85)/c3/{}/{left}/1/2/{below right}/-1}, {(2,3.85)/c4/x''/{above right}/-1/3/{below left}/-1}, {(-1,-0.85)/d1/y/{above left}/-1/1/below/1}, {(-3.68,-2.35)/d2/y'/{below left}/-1/2/{above}/1}, {(-4.68,-0.63)/d3/{}/{left}/1/3/left/1}, {(-2.68,-4.08)/d4/{}/{left}/1/1/{below left}/-1}, {(1,-0.85)/e1/z/{above right}/-1/3/below/1}, {(3.68,-2.35)/e2/z'/{below right}/-2/2/above/1}, {(4.68,-0.63)/e3/{}/{left}/1/1/right/1}, {(2.68,-4.08)/e4/{}/{left}/1/3/{below right}/-1}} 
       { \node[vertex] (\name) at \pos {};
       \node [\posn=\dist] at (\name) {$\label$};
       \node [\posnb=\distb] at (\name) {$\labelb$};
            }
\foreach \source/ \dest  in {c1/c2,  d1/d2,  e1/e2,  e1/c1, c1/d1, d1/e1, d2/d4, e2/e4, c2/c3, c2/c4, d2/d3, e2/e3}
       \path[edge,  very thick] (\source) --  (\dest);

\draw[edge, very thick] (c3) .. controls (0,5.2)  ..  (c4);
\draw[edge, very thick] (d3) .. controls (-5.24,-3.35)  ..  (d4);
\draw[edge, very thick] (e3) .. controls (5.24,-3.35) .. (e4);
 
\draw[edge, thick, dashdotted] (c2) .. controls (-5, 4) and (-7.04, 3.2) ..  (d2) node [midway, above, sloped] (TextNode) {$(2,3)$-path};
\draw[edge, thick, dashdotted] (c2) .. controls (5, 4) and (7.04, 3.2) ..  (e2) node [midway, above, sloped] () {$(1,2)$-path};
\draw[edge, thick, dashdotted] (d2) .. controls (-0.32, -7.9) and (0.32, -7.9) ..  (e2) node [midway, below, sloped] () {$(1,3)$-path};

\end{scope}

\end{tikzpicture}
\end{center}
\caption{Colourings of $G'$ in the proof of Lemma~\ref{l2}.  The dotted lines indicate paths of arbitrary length.}\label{fig:4}
\end{figure}
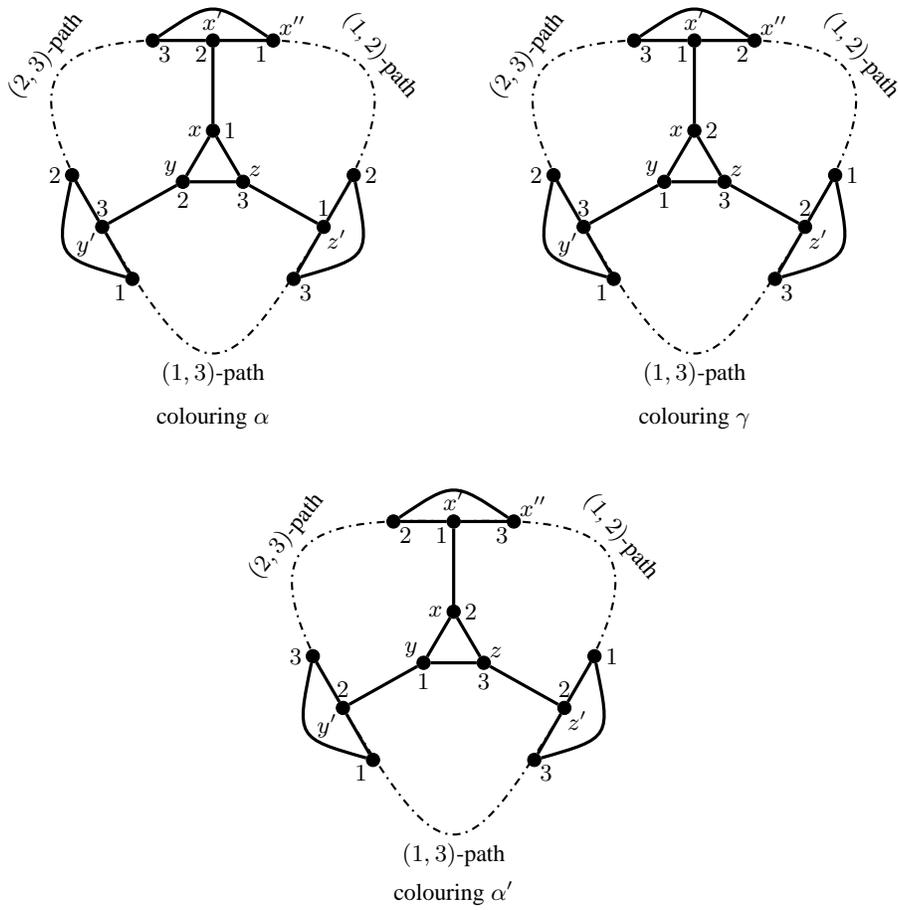

\subsection{Proof of Lemma~\ref{l3}}\label{a-c}

We first need another lemma.  

\begin{lemma} \label{lem-W}
Let $W$ be a set of three vertices in a 3-connected cubic graph $G$ such that every 3-colouring of $G$ colours at least two vertices of $W$ alike.  Let every 3-colouring~$c$ of $G$ be Kempe equivalent to another 3-colouring $c'$ such that $c$ and $c'$ colour alike distinct pairs of $W$.  Then $C_3(G)$ is a Kempe class.
\end{lemma}

\begin{proof}
Let $\alpha$ and $\beta$ be two 3-colourings of $G$.  To prove the lemma we show that \mbox{$\alpha \sim_3 \beta$}.  By Lemma~\ref{lidentify1}, it is sufficient to find a matching pair of colourings that are Kempe equivalent to $\alpha$ and $\beta$ respectively.

Let $W=\{x,y,z\}$.  We can assume that $\alpha(x)=\alpha(y)$.  If $\beta(x)=\beta(y)$, then $\alpha$ and $\beta$ match and we are done. So we can instead assume that $\beta(y)=\beta(z)$.  Let $\alpha'$ be a 3-colouring Kempe equivalent  to $\alpha$ that colours alike a different pair of  $W$.   If $\alpha'(y) = \alpha'(z)$, then $\alpha'$ and $\beta$ match.  Otherwise we must have that $\alpha'(x)=\alpha'(z)$.   Let~$\beta'$ be a 3-colouring Kempe equivalent  to $\beta$ that colours alike a different pair of  $W$.   So $\beta'(x) \in \{\beta'(y), \beta'(z) \}$ and $\beta'$ matches either $\alpha$ or $\alpha'$. \qed
\end{proof}

We restate Lemma~\ref{l3} before we present its proof.

\medskip
\noindent
{\bf Lemma~\ref{l3} (restated).}
{\it If $G$ is a $3$-connected cubic graph that is not claw-free then $C_3(G)$ is a Kempe class.}

\medskip

\noindent
{\it Proof.}
Note that if a vertex has three neighbours coloured alike it is a single-vertex Kempe chain.  We will write that such a vertex can be recoloured to refer to the exchange of such a chain.  

We make repeated use of Lemma~\ref{lidentify1}: two colourings are Kempe equivalent if they match.

Let $C$ be a claw in $G$ with vertex labels as in Figure~\ref{fig:1}.  Note that in every 3-colouring of $G$, two of $s$, $u$ and $v$ are coloured alike since they have a common neighbour.  If some fixed pair of $u$, $v$ and $s$ is coloured alike by every 3-colouring of $G$, then every pair of colourings matches and we are done. So let $\alpha$ be a 3-colouring of $G$ and assume that $\alpha(u)=\alpha(v)=1$ and that there are colourings for which $u$ and $v$ have distinct colourings, or, equivalently, colourings for which $s$ has the same colour as either $u$ or~$v$.  By Lemma~\ref{lem-W}, it is sufficient to find such a 3-colouring  that is Kempe equivalent to~$\alpha$. Our approach is to divide the proof into a number of cases, and, in each case, start from~$\alpha$ and make a number of Kempe changes until a colouring in which $s$ agrees with either $u$ or $v$ is obtained.  We will denote such a colouring $\omega$ to indicate a case is complete.

First some simple observations.  If $\alpha(s)=1$, then let $\omega=\alpha$ and we are done.  So we can assume instead that $\alpha(s)=2$ (and so, of course, $\alpha(w)=3$).    If it is possible to recolour one of $u$, $v$ or $s$, then we can let $\omega$ be the colouring obtained.  Thus we can assume now that each vertex of $u$, $v$ and $s$ has two  neighbours that are not coloured alike.

For a colouring $c$, vertex $x$, and colours $a$ and $b$ let $F_{c,x}^{ab}$ denote the $(a, b)$-component at $s$ under $c$.  We can assume that $F_{\alpha,s}^{12}$ contains both $u$ and $v$ as otherwise exchanging $F_{\alpha,s}^{12}$ results in a colouring in which~$s$ agrees with either $u$ or $v$.   

Let $N(u) = \{w, u_1, u_2\}$, $N(v) = \{w, v_1, v_2\}$, and $N(s) = \{w, s_1, s_2\}$.  Note that the vertices $u_1, u_2, v_1, v_2, s_1, s_2$ are not necessarily distinct.  

    \medskip
 \noindent
 {\bf Case 1.} $\alpha(u_1)\not=\alpha(u_2)$, $\alpha(v_1)\not=\alpha(v_2)$ and $\alpha(s_1)\not=\alpha(s_2)$. \\
   \noindent So each of $u$, $v$ and $s$ has degree 1 in $F_{\alpha,s}^{12}$ and therefore $F_{\alpha,s}^{12}$ has at least one vertex of degree 3.  Let $x$ be the vertex of degree 3 in $F_{\alpha,s}^{12}$ that is closest to $u$ and let $\alpha'$ be the colouring obtaining by recolouring $x$.  Then $u$ is not in $F_{\alpha',s}^{12}$ which can be exchanged to obtain $\omega$.

    \medskip
 \noindent
 {\bf Case 2.} $\alpha(s_1) = \alpha(s_2)$.  \\
   \noindent   Then $\alpha(s_1) = \alpha(s_2) = 1$ else $\omega$ can be obtained by recolouring $s$.
   
   \smallskip
   \noindent
   \textbf{Subcase 2.1:} $\alpha(u_1) = \alpha(u_2)$ or $\alpha(v_1) = \alpha(v_2)$.  \\
   \noindent   The two cases are equivalent so we consider only the first.  We have $\alpha(u_1) = \alpha(u_2) = 2$ else $u$ is not in $F_{\alpha,s}^{12}$.  Note that $F_{\alpha,s}^{23}$ consists only of $s$ and $w$. If $F_{\alpha,s}^{23}$ is exchanged, $u$ has three neighbours coloured $2$, and can be recoloured to obtain $\omega$ (as $u$ and $s$ are both now coloured 3).
   
   \smallskip
   \noindent
   \textbf{Subcase 2.2:} $\alpha(u_1) \not= \alpha(u_2)$ and $\alpha(v_1)\not=\alpha(v_2)$. \\
   \noindent We can assume that $\alpha(u_1) = \alpha(v_1)= 2$, and $\alpha(u_2) = \alpha(v_2) = 3$.  

In this case, we take a slightly different approach.  Let $\omega$ now be some fixed 3-colouring  with $\omega(s) \in \{\omega(u),\omega(v)$\}.We show that $\alpha \sim_3 \omega$ by making Kempe changes from $\alpha$ until a colouring that matches $\omega$ (or a colouring obtained from $\omega$ by a Kempe change) is reached.  

Let $\{a,b,c\}=\{1,2,3\}$. If $\omega(s_1)=\omega(s_2)$, then $\omega$ matches $\alpha$ (recall $\alpha(s_1)=\alpha(s_2)$ in this case). So assume that $\omega(s_1)=a$ and $\omega(s_2)=b$.  Then $\omega(s)=c$, and we can assume, without loss of generality, that $\omega(w)=a$.  Note that we can assume that $\mbox{$\omega(u) \neq \omega(v)$}$ else $\alpha$ and $\omega$ match and we are done.  So, as $u$ and $v$ are symmetric under~$\alpha$, we can assume that $\omega(u)=b$ and $\omega(v)=c$. If $\omega(u_2)=a$ or $\mbox{$\omega(v_2)=a$}$, then, again, $\alpha$ and $\omega$ match  (recall that $\alpha(w) = \alpha(u_2) = \alpha(v_2)$) so we assume otherwise (noting that this implies $\omega(u_2)=c$ and $\omega(v_2)=b$) and consider two cases.  For convenience, we first illustrate our current knowledge of $\alpha$ and $\omega$ in Figure~\ref{fig:2}.  (Though it is not pertinent in this case, we again observe that the six vertices of degree 1 in the illustraton might not, in fact, be distinct.)

\tikzstyle{vertex}=[circle,draw=black,  minimum size=13pt, inner sep=0pt]
\tikzstyle{edge} =[draw,-,black,>=triangle 90]

\begin{figure}
\begin{center}
\begin{tikzpicture}[scale=0.65]

\node at (0,-4.15) {colouring $\alpha$};
\node at (10,-4.15) {colouring $\omega$};

 \begin{scope}[xshift=0cm]
   \foreach \pos/\name / \label / \posn / \dist / \vertexlabel in {{(0,3)/c1/1/above/5/u}, {(0,0)/c2/3/{above right}/3/w}, {(2.6,-1.5)/c3/1/{right}/5/v}, {(-2.6,-1.5)/c4/2/{left}/5/s}, {(-1.74,3)/c5/2/{above left}/3/{u_1}}, {(1.74,3)/c6/3/{above right}/3/{u_2}}, {(-1.74,-3)/c7/1/{above right}/3/{s_2}}, {(-3.46,0)/c8/1/{above right}/3/{s_1}}, {(1.74,-3)/c9/2/{above left}/3/{v_1}}, {(3.46,0)/c10/3/{above left}/3/{v_2}}}
       { \node[vertex] (\name) at \pos {$\vertexlabel$};
       \node [\posn=\dist] at (\name) {$\label$};
       }

\foreach \source/ \dest  in {c2/c1,  c2/c4, c2/c3, c1/c5, c1/c6, c4/c7, c4/c8, c3/c9, c3/c10}
       \path[edge, black!50!white,  thick] (\source) --  (\dest);
       
\end{scope}

 \begin{scope}[xshift=10cm]
   \foreach \pos/\name / \label / \posn / \dist / \vertexlabel in {{(0,3)/c1/b/above/5/u}, {(0,0)/c2/a/{above right}/3/w}, {(2.6,-1.5)/c3/c/{right}/5/v}, {(-2.6,-1.5)/c4/c/{left}/5/s}, {(-1.74,3)/c5/{}/{above left}/3/{u_1}}, {(1.74,3)/c6/c/{above right}/3/{u_2}}, {(-1.74,-3)/c7/b/{above right}/3/{s_2}}, {(-3.46,0)/c8/a/{above right}/3/{s_1}}, {(1.74,-3)/c9/{}/{above left}/3/{v_1}}, {(3.46,0)/c10/b/{above left}/3/{v_2}}}
       { \node[vertex] (\name) at \pos {$\vertexlabel$};
       \node [\posn=\dist] at (\name) {$\label$};
       }

\foreach \source/ \dest  in {c2/c1,  c2/c4, c2/c3, c1/c5, c1/c6, c4/c7, c4/c8, c3/c9, c3/c10}
       \path[edge, black!50!white,  thick] (\source) --  (\dest);
      
\end{scope}

\end{tikzpicture}
\end{center}
\caption{The colourings of Subcase 2.2 of Lemma~\ref{l3}.}\label{fig:2}
\end{figure}
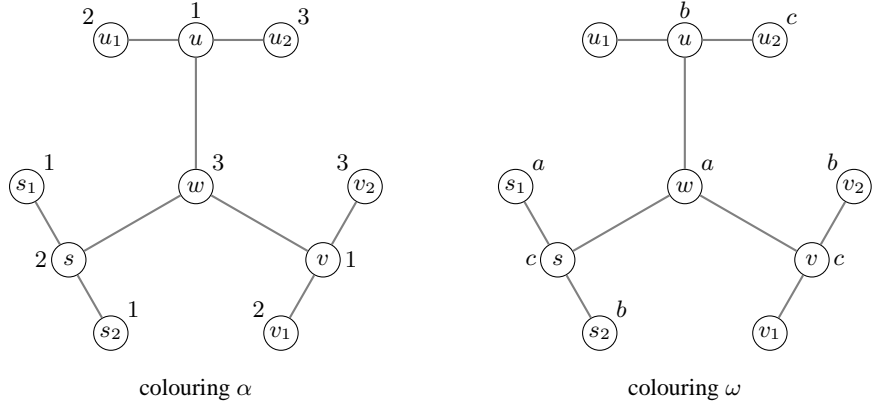

   \smallskip 
   \noindent
   \textbf{Subcase 2.2.1:} $\omega(w) =a  \in \{\omega(u_1), \omega(v_1) \}$. \\
    \noindent Notice that $F_{\alpha,s}^{23}$ contains only $s$ and $w$.  If it is exchanged then a colouring is obtained where $w$, $u_1$ and $v_1$ are coloured alike and this colouring matches $\omega$.
    
    \smallskip 
    \noindent
    \textbf{Subcase 2.2.2:} $\omega(w) =a \not\in \{\omega(u_1), \omega(v_1) \}$. \\
     \noindent So $\omega(u_1)=c$ and $\omega(v_1)=b$.  Thus $F_{\omega,w}^{ab}$ contains only $u$ and $w$, and the colouring obtained by its exchange matches $\alpha$ as $w$ and $v_1$ are both coloured $b$.

     \medskip
 \noindent
 {\bf Case 3.} $\alpha(u_1) = \alpha(u_2)$, $\alpha(v_1) \not= \alpha(v_2)$, and $\alpha(s_1)\not= \alpha(s_2)$. \\
     \noindent   If $\alpha(u_1) = \alpha(w)$ then the three neighbours of $u$ are coloured alike and it can be recoloured to obtain $\omega$.  So suppose $\alpha(u_1) = \alpha(u_2)= 2$. We may assume that $\alpha(s_1) = 1$, $\alpha(s_2) = \alpha(v_2) = 3$, and $\alpha(v_1) = 2$; see the illustration of Figure~\ref{fig:3}.

\tikzstyle{vertex}=[circle,draw=black,  minimum size=13pt, inner sep=0pt]
\tikzstyle{edge} =[draw,-,black,>=triangle 90]

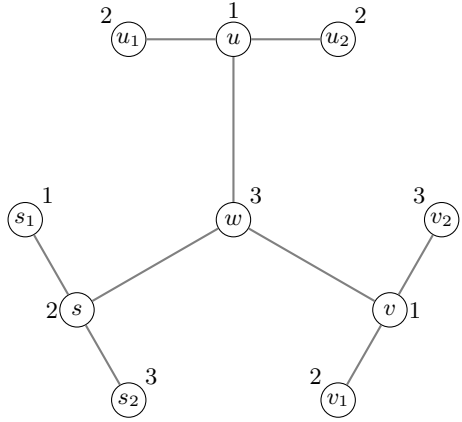
\begin{figure}
\begin{center}
\begin{tikzpicture}[scale=0.8]

   \foreach \pos/\name / \label / \posn / \dist / \vertexlabel in {{(0,3)/c1/1/above/5/u}, {(0,0)/c2/3/{above right}/3/w}, {(2.6,-1.5)/c3/1/{right}/4/v}, {(-2.6,-1.5)/c4/2/{left}/4/s}, {(-1.74,3)/c5/2/{above left}/3/{u_1}}, {(1.74,3)/c6/2/{above right}/3/{u_2}}, {(-1.74,-3)/c7/3/{above right}/3/{s_2}}, {(-3.46,0)/c8/1/{above right}/3/{s_1}}, {(1.74,-3)/c9/2/{above left}/3/{v_1}}, {(3.46,0)/c10/3/{above left}/3/{v_2}}}
       { \node[vertex] (\name) at \pos {$\vertexlabel$};
       \node [\posn=\dist] at (\name) {$\label$};
       }
\foreach \source/ \dest  in {c2/c1,  c2/c4, c2/c3, c1/c5, c1/c6, c4/c7, c4/c8, c3/c9, c3/c10}
       \path[edge, black!50!white,  thick] (\source) --  (\dest);
       
\end{tikzpicture}
\end{center}
\caption{The colouring $\alpha$ of Case 3 of Lemma~\ref{l3}.}\label{fig:3}
\end{figure}

\smallskip

\noindent We continue to assume that $F_{\alpha,s}^{12}$ contains $u$ and $v$ and note that $s$ and $v$ have degree 1 therein. 

    \smallskip 
    \noindent
    \textbf{Subcase 3.1:} $F_{\alpha,s}^{12}$ is not a path. \\
\noindent  Let $t$ be vertex of degree $3$ closest to $s$ in $F_{\alpha,s}^{12}$.  Then $t$ can be recoloured to obtain a colouring $\alpha'$ such that $F_{\alpha',s}^{12}$ does not contain $v$.  Exchanging $F_{\alpha',s}^{12}$, we obtain $\omega$.

    \smallskip 
    \noindent
    \textbf{Subcase 3.2:} $F_{\alpha,s}^{12}$ is a path. \\
\noindent  Note that $F_{\alpha,s}^{12}$ is a path from $s$ to $v$ through $s_1$ and $u$.

    \smallskip 
    \noindent
    \textbf{Subcase 3.2.1:} $F_{\alpha,s_2}^{13}$ is a path from $s_1$ to $s_2$. \\
\noindent  Note that $F_{\alpha,u}^{13} \neq F_{\alpha,s_2}^{13}$ since if $F_{\alpha,u}^{13}$ is a path then $u$ would be an endvertex coloured~1 implying $u=s_1$ contradicting that $C$ is a claw.  As $G$ is cubic a vertex can belong to both $F_{\alpha,s}^{12}$ and $F_{\alpha,s_2}^{13}$ if it is an endvertex of one of them, and we note that $s_1$ is the only such vertex.  
     
Let $\alpha'$ be the colouring obtained from $\alpha$ by the exchange of $F_{\alpha,s_2}^{13}$.  If $s \not\in F_{\alpha',v}^{12}$, then let $\omega$ be the colouring obtained by the further exchange of $F_{\alpha',v}^{12}$.

Otherwise, $F_{\alpha',v}^{12}=F_{\alpha',s}^{12}$, $s$ and $v$ each have degree 1 therein, and we can assume it is a path (else, as in Subcase 3.1, there is a vertex of degree 3 that can be recoloured to obtain $\alpha''$ and $F_{\alpha'',s}^{12}$ does not contain $v$ and can be exchanged to obtain $\omega$).  We can also assume that $F_{\alpha',s}^{12}$ contains $F_{\alpha,s}^{12} \setminus \{s_1\}$:  if not, then $F_{\alpha, s_2}^{13} \backslash \{s_1, s_2\} \cap F_{\alpha, v}^{12} \not= \emptyset$ (recall that $F_{\alpha,s}^{12}$ is a path from $s$ to $v$ through $s_1$ and $u$) but their common vertices would have degree $4$. Thus, in particular, $F_{\alpha',s}^{12}$ contains $u$ and the vertex~$t$ at distance 2 from $s$ in~$F_{\alpha,s}^{12}$.

As $t$ is not an endvertex in $F_{\alpha',s}^{12}$, $s_1$ is its only neighbour coloured 3 under $\alpha'$.  So $F_{\alpha',w}^{23}$ contains four vertices: $w$, $s$, $s_1$ and $t$. Let $\alpha''$ be the colouring obtained from $\alpha'$ by the exchange of $F_{\alpha',w}^{23}$. If $t \not\in \{u_1, u_2\}$, then $u$ has three neighbours with colour $2$ with $\alpha''$ and so can be recoloured to obtain $\omega$.  Otherwise the conditions of  Case 1 are now met.
 
    \smallskip 
    \noindent
    \textbf{Subcase 3.2.2:} $F_{\alpha,s_2}^{13}$ is not a path from $s_1$ to $s_2$. \\
\noindent  If $s_1 \notin F_{\alpha,s_2}^{13}$, then the exchange of $F_{\alpha,s_2}^{13}$ gives a colouring in which $s_1$ and $s_2$ are coloured alike (the colour of $s$ is not affected by the exchange and either both or neither of $u$ and $v$ change colour).  Thus Case 2 can now be used.

So we can assume that $s_1 \in F_{\alpha,s_2}^{13}$ has degree 1 in $F_{\alpha,s_2}^{13}$ (recall that $s_1$ has degree~2 in $F_{\alpha,s}^{12}$).  If $s_2$ has degree 2 in $F_{\alpha,s_2}^{13}$, then $F_{\alpha, s}^{23}$ contains only $w$, $s$ and $s_2$.  If it is exchanged, $u$ has three neighbours with colour $2$ and can be recoloured to $\omega$.  

Thus $s_1$ and $s_2$ both have degree 1 in $F_{\alpha,s_2}^{13}$.  Let $x$ be the vertex of $F_{\alpha,s_2}^{13}$ closest to~$s_2$. Then $x$ can be recoloured to obtain a colouring $\alpha'$ such that $F_{\alpha',s_2}^{13}$  does not contain $s_1$.  Exchanging $F_{\alpha',s_2}^{13}$ again takes us to Case 2. This completes Case 3.

By symmetry, we are left to consider the following case to complete the proof of the lemma.

     \medskip
 \noindent
 {\bf Case 4.} $\alpha(u_1) = \alpha(u_2)$, $\alpha(v_1) = \alpha(v_2)$, and $\alpha(s_1)\not= \alpha(s_2)$. \\ 
     \noindent   If $\alpha(v_1) = \alpha(v_2) = 3$, then $v$ can be recoloured to obtain $\omega$.  So we can assume that  $\alpha(v_1) = \alpha(v_2) = 2$, and, similarly, that $\alpha(u_1) = \alpha(u_2) = 2$.  We can also assume that $F_{\alpha,s}^{23}$ is a path since otherwise the vertex of degree 3 closest to $s$ can be recoloured.  Define $S = \{u_1, u_2, v_1, v_2\}$. We distinguish two cases.   
     
     \smallskip
     \noindent
     \textbf{Subcase 4.1:} $|S \cap F_{\alpha,s}^{23}| \geq 2$. \\
     \noindent     As $F_{\alpha,s}^{23}$ is a path and $w$ is an endvertex, one vertex of $S$, say $v_1$, has degree $2$ in $F_{\alpha,s}^{23}$. Consider $F_{\alpha,w}^{13}$: it consists only of vertices $w$, $u$, and $v$.  After it is exchanged, $v_1$ has three neighbours with colour $3$ and recolouring $v_1$ allows us to apply Case 3. 
     
     \smallskip
     \noindent
     \textbf{Subcase 4.2:} $|S \cap F_{\alpha,s}^{23}| \leq 1$. \\
    \noindent    It follows, without loss of generality, that $\{u_1, u_2\} \cap F_{\alpha,s}^{23} = \emptyset$.  Exchange $F_{\alpha, u_1}^{23}$and $F_{\alpha, u_2}^{23}$ (which might be two distinct components or just one) to obtain a colouring $\alpha'$.  As $w \in F_{\alpha',s}^{23}$ (and hence $w \not\in F_{\alpha, u_1}^{23} \cup F_{\alpha, u_2}^{23}$), every neighbour of $u$ is coloured 3 and it can be recoloured to obtain $\omega$. This completes Case 4 and the proof of Lemma~\ref{l3}.
     \qed

\bibliography{bibliography}{}
\bibliographystyle{abbrv}

\end{document}